\documentclass[twocolumn]{aastex631}
\usepackage{physics,amsmath,amssymb,bbold,amsthm,mathrsfs}
\usepackage{xcolor,soul}

\newtheorem{theorem}{Theorem}

\newcommand{\Change}[1]{{\color{red}\bf #1}}
\renewcommand{\Change}[1]{#1}

\newcommand{\mki}{
  Kavli Institute for Astrophysics and Space Research, 
  Massachusetts Institute of Technology , 77 Massachusetts  Ave., 
  Cambridge, MA 02139, USA
}

\begin{document}

\title{Unfolding X-ray Spectral Data: Conditions and Applications}

\author[0000-0003-2602-6703]{Sean J. Gunderson}\thanks{seang97@mit.edu}\affiliation{\mki}

\author[0000-0002-3860-6230]{David P.\ Huenemoerder}\affiliation{\mki}

\begin{abstract}

We present conditions for which X-ray spectra can be ``unfolded'' to present accurate \Change{representation} of the true source spectra. The method we use to unfold the data is implemented in the \textit{Interactive Spectral Interpretation Software} \citep{Houck2000} and distinguishes itself as being model-independent. We find that this method of unfolding makes accurate representations of the true source spectra \Change{(1) The detector is high-resolution and (2) The spectrum is not steeply sloped. These criteria are not simple conditions that give concrete determinations; each detector and spectrum must be judged individually. We find that both grating and imaging detectors can be unfolded with minimal distortions as compared to both continuum and local spectral features; the latter CCD detectors being much more energy dependent.} We also provide example use cases for unfolding in the context of current generation X-ray observatories and important caveats.
\end{abstract}


\section{Introduction} \label{sec:intro}

As defined in \citet{Davis2001}, the observed count rate within an X-ray detector's channel $\lambda'$ from a 
source \Change{with spectral energy distribution} $s(\lambda)$ ($\mathrm{phot\,cm^{-2}s^{-1}\AA^{-1}}$)  is
\begin{equation}
    C(\lambda') = t \int_0^\infty s(\lambda) \mathscr{R}(\lambda', \lambda)\dd\lambda,\label{eq:CountsDef}
\end{equation}
where $\mathscr{R}(\lambda', \lambda)$ is the response of the detector and $t$ is the exposure time of the 
observation. The integration is over all incident photon wavelengths $\lambda$, though in practice it is 
truncated
to the range that defines the detector's sensitivity. As written, the response function $\mathscr{R}(\lambda', 
\lambda)$ combines both the effective area $A(\lambda)$ ($\mathrm{cm^2\, counts/photon}$) and redistribution 
function $R(\lambda', \lambda)$ (unitless). 
\Change{The effective area is the geometric area the optics would have if the reflection, transmission, and 
detection 
efficiencies were perfect for an incident photon of wavelength $\lambda$.  The redistribution function describes the smearing of the input photon across 
multiple detector channels, either spatially for dispersive spectrometers, or into arbitrary detector channels 
for non-dispersive spectrometers, like CCDs or proportional counters.} 
\citet{Davis2001} showed that for grating spectrometers, it is more appropriate to define $\mathscr{R}
(\lambda',\lambda)=A(\lambda',\lambda)R(\lambda)$ where as CCD detectors are \Change{more easily described by} 
$\mathscr{R}(\lambda',\lambda)=A(\lambda)R(\lambda',\lambda).$ \Change{However, as \mbox{\citet{Davis2001}} noted, 
the 
choice is arbitrary,} so without loss of generality, we will use the latter formulation, \Change{since this is 
used in  current software implementations and in standard calibration file formats}.\footnote{Even though it 
is 
a grating instrument, the \textit{Chandra} High Energy Transition Grating System's response files are defined 
like 
that of a CCD. Our choice is thus to also make comparisons to real instruments easier.}

Equation~\eqref{eq:CountsDef} is the basis of X-ray astronomy and the forward folding technique 
\citep{Arnaud1996} used extensively. 
\Change{One chooses a plausible model, $s_M(\lambda;\textbf{q})$ where $\textbf{q}$ is an array of model parameters, and minimizes an appropriate statistic, such as a $\chi^2$ (or one appropriate for a 
low-counts regime), by adjusting those parameters.  This iterative method is used because 
Equation~\eqref{eq:CountsDef} 
cannot be mathematically inverted for realistic response functions.  One must use physical intuition about the 
plausible forms of the solution and then see if there are any parameters which give model counts statistically 
consistent with the observed counts. The model $s_M(\lambda;\textbf{q})$ that produces a count spectra that is
statistically consistent with the observation $C(\lambda')$ is not a
determination of the true source, but is an estimate for physically useful parameters that can
describe the system. As a consequence, many different models can also fit a given spectra and describe the
source.} The forward folding method has seen broad success, especially as more 
efficient minimization techniques have followed advances in X-ray telescope resolutions.

Attempts at solving for the true source $s(\lambda')$ have occurred concurrently 
\citep{Blissett1979,Loredo1989,Rhea2021}, but these methods are not always practical for two reasons. First, 
X-ray detector responses are non-diagonal, posing computational limitations. Secondly, and more importantly, 
any solution to Equation~\eqref{eq:CountsDef} will not be unique. From Fredholm's theorem on 
integral equations, there will be a set of \Change{$n$-number of} linearly-independent solutions $\{s_1(\lambda),...,s_n(\lambda)\}$. In practice, this 
means that for any observation $C(\lambda')$, there are $n$-many possible functions that can be called the 
``source.'' \Change{The set of solutions is further expanded when uncertainties in the observed
counts in any given channel are included; each combination of counts within error produces
$n$-number of solutions.}

\Change{The number of solutions grows even more when one considers the many details that are hidden in Equation~\eqref{eq:CountsDef} that make it look mathematically possible to solve. 
The most important of which is the meaning of the ``source'' $s(\lambda)$. The extracted count spectra $C(\lambda')$ is heavily dependent on 
the choices made in its extraction (e.g., source and background regions, good-time intervals, event detection 
thresholds, grating order sorting). This further complicates any solving of Equation~\eqref{eq:CountsDef} since 
a small change in extraction parameters will change the count spectra.}

\Change{This can be seen in the the full, formal expression for the counts spectra
\begin{equation}
    C(\lambda',\sigma,t) = \int_0^\infty \int_\Omega s(\lambda, \hat{\mathbf{p}}, t)\mathscr{R}(\lambda',\lambda,\sigma, \hat{\mathbf{p}}, t)\dd\hat{\mathbf{p}}\dd\lambda.
\end{equation}
from \citet{Davis2001}. The response and source are both dependent
on time $t$ and the direction of the incoming photon $\hat{\mathbf{p}}$ from a region on the sky $\Omega$, while the response is also
dependent on the location (area) $\sigma$ that the photon is incident on the detector. These are all
handled internally in the data reduction of the raw data when an extraction occurs, and can be simplified
as \citet{Davis2001} shows to get Equation~\eqref{eq:CountsDef}, but must be considered when solving for
$s(\lambda)$ explicitly. These additional extraction parameters create effectively an infinite number of solutions
for any given X-ray observation.

A priori knowledge can be used to reduce
the number of solutions to some smaller set, which has been used on some systems \citep{Blissett1979, Loredo1989}, but we do not always have prior knowledge of the physics
of a given system (see the discourse on $\gamma$\,Cas stars for example; \citealp{Langer2020}). Forward folding is required in these situations despite not providing a true source flux model. Yet, having some form of a \Change{model-independent source spectrum} is the ``holy grail'' of X-ray astronomy and why efforts to de-convolve an observed spectra $C(\lambda')$ continue. Note, though, that we have only stated that solving for $s(\lambda)$ in a functional form is impossible (or at least not practical). If our goal is only to have a \textit{representation} of the source $\bar{S}(\lambda')$, say by ``unfolding'' it from the response, then we circumvent the Fredholm problem since we are not trying to find a functional form of $s(\lambda)$.

}  

We will define this representation $\bar{S}(\lambda')$ as the the unfolded flux spectra
\begin{equation}
    \bar{S}(\lambda')\equiv \frac{C(\lambda')}{t\int_{\Delta \lambda(\lambda')} \mathscr{R}(\lambda', \lambda)\dd\lambda}.\label{eq:Sbar}
\end{equation}
\Change{The denominator represents the counts distribution one would get from a source spectrum 
with unit flux, $s(\lambda)=1$.  We are dividing the observed counts by the convolved counts per 
unit flux, and hence the result has units of flux.}
By unfolding, we are attempting to remove the effects of the response from the observed counts in 
a model-independent fashion. The denominator is the characteristic response of the detector 
\Change{smoothed over all wavelengths in $\Delta\lambda$ that are read as $\lambda'$, and scaled 
by the exposure time.} 

\Change{Note that while we have been defining} $C(\lambda')$ as the count spectrum of an 
observation, it 
is more formally the background subtracted count spectrum $C(\lambda') = c(\lambda') - 
B(\lambda')$. This is in general a trivial inclusion to the problem since the background only 
increases the errors in the count spectra. There may be some cases where the background can not be 
simply unfolded along with the source counts, but these cases are few in number and beyond the 
scope of this work.

We are not attempting to take credit for the algorithms or inception of this unfolding method. It 
was originally implemented in the \textit{Interactive Spectral Interpretation System} 
(\textsc{isis}; \citealp{Houck2000})\footnote{\url{https://space.mit.edu/cxc/isis/}}. Our goal is 
to make this type of unfolding more widely known and explore the limits of when 
Equation~\eqref{eq:Sbar} is an accurate representation of $s(\lambda)$. As a first step to doing 
this, we first need to explore important properties of Equation~\eqref{eq:Sbar} and how it 
compares to other methods of unfolding.

\subsection{Important Theorems}

While we have said we want to accurately represent of $s(\lambda)$, it is more accurate to say 
that we want to compare an imperfect response $\mathscr{R}(\lambda,\lambda')$ to an idealized, 
perfect response
\begin{equation}
    \mathscr{R}_\mathrm{P}(\lambda',\lambda) = A(\lambda)\delta(\lambda-\lambda')\label{eq:PerfectResponse}
\end{equation}
where $\delta(\lambda-\lambda')$ is the Dirac delta function. We will assume here that $A(\lambda)$ is the area factor for both $\mathscr{R}_\mathrm{P}$ and $\mathscr{R}$ such that the detectors are equivalent in all but their resolving power. Note that for a perfect detector, we can map every incident photon exactly to a corresponding $\lambda'$, meaning $\lambda' = \lambda$. \Change{In other words, the detector has continuous bins \citep{Davis2001}, meaning infinite resolving power, so every incident photon's wavelength is known exactly.} Given the perfect response, the true flux spectra is then
\begin{equation}
    S(\lambda') = \frac{\int_0^\infty s(\lambda)\mathscr{R}_\mathrm{P}(\lambda',\lambda)\dd\lambda}{\int_{\Delta\lambda(\lambda')} \mathscr{R}_\mathrm{P}(\lambda',\lambda)\dd\lambda} = \frac{A(\lambda')s(\lambda')}{A(\lambda')}=s(\lambda')\label{eq:Strue}
\end{equation}
This perfectly-redistributed spectra $S(\lambda')$ will be the basis of comparison against the unfolded spectra $\bar{S}(\lambda')$. To start, we first need to address the uniqueness problem.

\begin{theorem}
    For a given X-ray detector response \mbox{$\mathscr{R}(\lambda',\lambda)$} used to make an observation of a source $s(\lambda)$, the unfolded spectra $\bar{S}$ given by Equation~\eqref{eq:Sbar} is unique.
\end{theorem}

\begin{proof}
    Assume by way of contradiction that there exists two unfolded spectra $\bar{S}_1(\lambda')$ and $\bar{S}_2(\lambda')$ of a single source $s(\lambda)$ such that $\bar{S}_1(\lambda')\neq \bar{S}_2(\lambda')$ for some $\lambda'$. In one of these channels, we can then say from Equations~\eqref{eq:CountsDef} and \eqref{eq:Sbar} that
    \begin{equation}
        \int_0^\infty s(\lambda)\mathscr{R}(\lambda',\lambda)\dd\lambda\neq\int_0^\infty s(\lambda)\mathscr{R}(\lambda',\lambda)\dd\lambda
    \end{equation}
    Which implies that $s(\lambda)\neq s(\lambda)$ in the given channels. This contradicts our assumption that we only have one source. Thus there can not be two different unfolded spectra for the same \Change{extracted spectrum}.
\end{proof}

An alternative unfolding is provided in the \textsc{xspec} \citep{Arnaud1996} as
\begin{equation}
    \bar{S}_\textsc{xspec}(\lambda') = C(\lambda')\frac{s_M(\lambda')}{t\int_0^\infty s_M(\lambda)\mathscr{R}(\lambda',\lambda)\dd\lambda},\label{eq:XspecUnfold}
\end{equation}
where $s_M(\lambda)$ is an assumed model. \Change{This method of unfolding spectra can provide a model-independent representation of the source, but \textit{only} in the case of a perfect detector. If we insert Equations~\eqref{eq:CountsDef} and \eqref{eq:PerfectResponse} into \eqref{eq:XspecUnfold}, we see that
\begin{eqnarray}
    \bar{S}_\textsc{xspec}(\lambda') &&= \frac{s_M(\lambda')\int_0^\infty s(\lambda)A(\lambda)\delta(\lambda-\lambda')\dd\lambda}{\int_0^\infty s_M(\lambda) A(\lambda) \delta(\lambda-\lambda')\dd\lambda}\nonumber\\
    &&= \frac{s_M(\lambda')s(\lambda')A(\lambda')}{s_M(\lambda')A(\lambda')}\nonumber\\
    &&= s(\lambda')
\end{eqnarray}
If the response shows any deviations from the perfect response, even if the resolving power is exceptionally high, the \textsc{xspec} version of unfolding will give a model-\textit{dependent} representation. While this may have some historic uses, it is not the objective we have outlined that unfolding should produce.}

The \textsc{isis} method does not have this problems and will be shown to be particularly good at reproducing the source flux\Change{, even for imperfect, lower resolution detectors}. This is not to say that the \textsc{isis} method is without problems. Even when we know the true source $s(\lambda)$, there is some non-zero error when unfolding. More importantly, for the cases where we do not know $s(\lambda)$, the error is dependent on the assumed model.

\begin{theorem}
    The error in unfolding the count spectrum using Equation~\eqref{eq:Sbar} depends on the assumed model and is always non-zero for some $\lambda'$.
\end{theorem}

\begin{proof}
    There are two cases we will consider. First, assume we know exactly the true source $s(\lambda)$. In this case, the error due to unfolding in a given channel $\varepsilon(\lambda')$ is
    \begin{equation}
        \varepsilon(\lambda')=\left|\frac{\bar{S}(\lambda')-S(\lambda')}{S(\lambda')}\right|.\label{eq:Uncertainty}
    \end{equation}
    \Change{Applying Equations~\eqref{eq:CountsDef}, \eqref{eq:Sbar}, and \eqref{eq:Strue} to Equation~\eqref{eq:Uncertainty}}, the error is more explicitly
    \begin{equation}
        \varepsilon(\lambda')=\left|\frac{\int_0^\infty s(\lambda) \mathscr{R}(\lambda', \lambda)\dd\lambda}{s(\lambda')\int_{\Delta \lambda(\lambda')} \mathscr{R}(\lambda', \lambda)\dd\lambda} - 1\right|.
    \end{equation}
    Now, we again we have two cases. First, if the source is constant for some small range of $\lambda$, the uncertainties go to zero for the channels that correspond to those input wavelengths. In the second case of when $s(\lambda)$ is non-constant, the error will be non-zero for the corresponding channels as long as the response function is imperfect.

    Next, let us assume we do not know the true source and are comparing against an assumed model $s_M(\lambda)$ such that $s_M(\lambda')\neq s(\lambda)$. The uncertainty in the channel will be
    \begin{equation}
        \varepsilon(\lambda')=\left|\frac{\int_0^\infty s(\lambda) \mathscr{R}(\lambda', \lambda)\dd\lambda}{s_M(\lambda')\int_{\Delta \lambda(\lambda')} \mathscr{R}(\lambda', \lambda)\dd\lambda} - 1\right|.
    \end{equation}
    Here we can see the  see the dependence on the assumed model $s_M(\lambda')$ and that regardless of whether $s(\lambda)$ is constant or not, the uncertainty will be non-zero. 
\end{proof}

It should be emphasized here that the uncertainty defined by Equation~\eqref{eq:Uncertainty} is on a channel-by-channel basis. How one quantifies the accuracy of unfolding is open to interpretation. In this work, we will use channel-by-channel comparisons, but one could also define single quantity measures as well, such as a \Change{$\chi^2$}.

This paper is organized as follows. In \S~\ref{sec:ToyModels}, we show the conditions underwhich spectra can be unfolded accurately. This is done for both grating and CCD detecors. In \S~\ref{sec:RealDetectors}, we give examples of how unfolding can be used on real X-ray data. Finally, in \S~\ref{sec:Conclusions}, we give our results and conclusions.

\section{Idealized Toy Models}\label{sec:ToyModels}

In this section we will illustrate the unfolding of spectra under different conditions. Each of the examples will be of idealized, imperfect detectors. The idealization will be made for the purposes of providing analytical formulae for the spectra. For real detector examples, see \S~\ref{sec:RealDetectors}

\subsection{Grating Detectors}

\subsubsection{Gaussian Response and Constant Area}

For this example, we will use a response function describing the constant width of the response 
for grating detectors. \Change{A common approximation of a grating detector's response function is 
a Gaussian {\citep[e.g.,][]{Herman2004b,Kaastra2016}}, so we will use such a functional form:}
\begin{equation}\label{eq:gaussrmf}
    \mathscr{R}(\lambda',\lambda) = A_0 \frac{R_0}{\sqrt{2\pi}\sigma} \exp(-\frac{1}
    {2}\left(\frac{\lambda'-\lambda}{\sigma}\right)^2),
\end{equation}
where $A_0$ is the effective area that we assume to be constant for all wavelengths, $R_0$ is the 
normalization of the detector, and $\sigma$ is the width of the response to an incoming photon 
$\lambda$.

We will first show the accuracy of unfolding a broadband continuum. For simplicity, we will \Change{assume we have observed a source with a power law flux density}
\begin{equation}
    s(\lambda) = s_0 \left(\frac{\lambda}{\lambda_0}\right)^\gamma,
\end{equation}\label{eq:plawsrc}
where $s_0$ is the flux density at $\lambda_0$. 
\Change{Substituting this source and the Gaussian response of Equation~\eqref{eq:gaussrmf} into Equations~\eqref{eq:CountsDef} and \eqref{eq:Sbar}, and performing the integration, we can derive the unfolded flux spectrum:}
\begin{eqnarray}
    \bar{S}(\lambda') = &&\frac{s_0 2^{\gamma/2}\sigma^{\gamma-1}}{\lambda_0^\gamma\sqrt{\pi}(1+\mathrm{erf}(\frac{\lambda'}{\sqrt{2}\sigma}))}\\
    &&\times\left(\sqrt{2}\lambda'\Gamma\left(1+\frac{\gamma}{2}\right){}_1F_1\left(\frac{1+\gamma}{2};\frac{3}{2};-\frac{1}{2}\left(\frac{\lambda'}{\sigma}\right)^2\right)\right.\nonumber\\
    &&+\left.\sigma\Gamma\left(\frac{1+\gamma}{2}\right){}_1F_1\left(-\frac{\gamma}{2};\frac{1}{2};-\frac{1}{2}\left(\frac{\lambda'}{\sigma}\right)^2\right)\right),\nonumber
\end{eqnarray}
where $\mathrm{erf}(x)$ is the error function, $\Gamma(z)$ is the complete gamma function, and ${}_1F_1(a;b;z)$ is the confluent hypergeometric function of the first kind. The error in the unfolding for this model is given in Figure~\ref{fig:GratingsPowerlawError} for varying values of $\gamma$ (left panel) and response width $\sigma$ (right panel).

\begin{figure*}
    \centering
    \begin{minipage}{0.49\linewidth}
        \includegraphics[width=\linewidth]{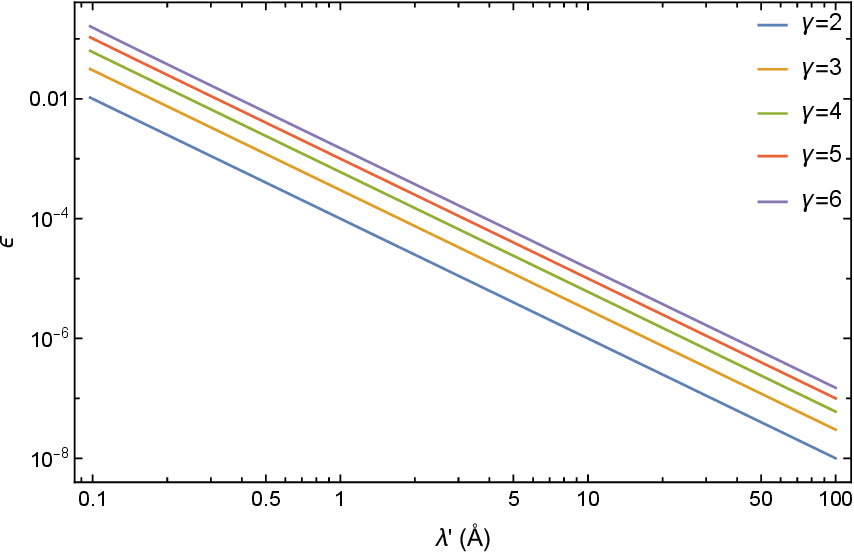}
    \end{minipage}
    \begin{minipage}{0.49\linewidth}
        \includegraphics[width=\linewidth]{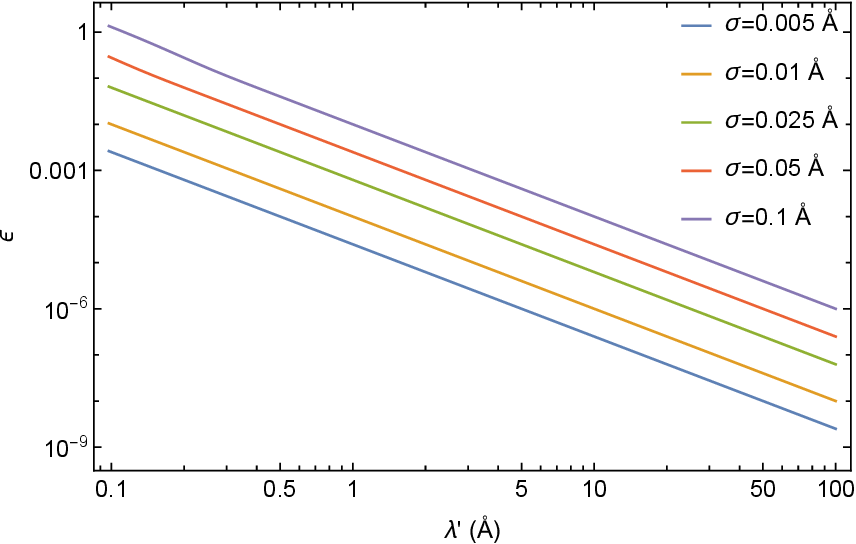}
    \end{minipage}
    \caption{Error in the unfolding of a power law continuum observed with a Gaussian response with constant effective area. The power law index is varied in the left panel while $\sigma=2$ is held constant. The right panel shows $\sigma$ varyind while $\gamma=2$ is held constant.}
    \label{fig:GratingsPowerlawError}
\end{figure*}

The negative correlation between the errors and channel wavelength is a feature of the grating detectors with constant response width, but it highlights how important a role the response plays. The more wavelengths covered by the response width, the worse the unfolding is. At the same time, the steepness of the incident flux spectra can compound on this problem as shown in the left panel of Figure~\ref{fig:GratingsPowerlawError}. The implications of this is that steeper the source feature at any given wavelength for a grating detector, the worse the uncertainty in the unfolding.

The width of the response has a much more significant effect on the unfolding error. For \textit{Chandra} grating widths (blue and orange curves), the errors are well below 1 percent for all wavelengths. As the resolution goes down, the shortest wavelengths can have significant error despite the longer wavelengths continuing to have minimal error. The wavelengths a detector is sensitive too can of course modify the determination of whether the continuum is well unfolded. Lower resolution instruments may be restricted to longer wavelengths where the unfolding is minimally-error prone (see \textit{XMM-Newton} below for example). We can conclude, generally speaking, that broadband continuum features can be unfolded with minimal distortion from the assumed source model if they are not significantly steep.

The next question is whether local features, like an emission line, get distorted by unfolding. For that we will use a Gaussian-line profile
\begin{equation}
    s(\lambda) = \frac{s_0}{\sqrt{2\pi}\sigma_L} \exp(-\frac{1}{2}\left(\frac{\lambda-\lambda_L}{\sigma_L}\right)^2),
\end{equation}
with line center $\lambda_L$ and width $\sigma_L$. \Change{Following the same procedure as with the power law, the unfolded line flux from our grating response is}
\begin{eqnarray}
    \bar{S}(\lambda') = &&\frac{s_0}{\sqrt{2\pi}}\frac{1}{\sqrt{\sigma^2+\sigma_L^2}}\frac{1+\erf\left(\frac{1}{\sqrt{2}}\frac{\lambda_L \sigma^2+\lambda'\sigma_L^2}{\sigma\sigma_L\sqrt{\sigma^2+\sigma_L^2}}\right)}{1+\erf\left(\frac{1}{\sqrt{2}}\frac{\lambda'}{\sigma}\right)}\nonumber\\
    &&\times\exp(-\frac{1}{2}\frac{(\lambda'-\lambda_L)^2}{\sigma^2+\sigma_L^2}).
\end{eqnarray}
For this case, comparing against the source Gaussian profile is not a useful metric due to the broadening caused by the response. Instead, we will compare the unfolded lines to the lines in counts space to look for any distortions in the line shape. This is done in two combinations in Figure~\ref{fig:GratingLineProfileError}. In the left panel, lines of similar width at different wavelengths (as appropriate for the \textit{Chandra} \Change{High Energy Transmission Grating System's (HETG) Medium Energy Grating (MEG)} pre-contamination; \citealp{Herman04}, \citealp{ODell17}) are unfolded and plotted with their count spectra version. Across all wavelengths in our toy model detector, the unfolding introduces no distortions to the shape. This seems obvious given the smoothness of the model response, but it highlights that dividing out the response in a small region will not distort the line.

\begin{figure*}
    \center
    \begin{minipage}{0.49\linewidth}
        \includegraphics[width=\linewidth]{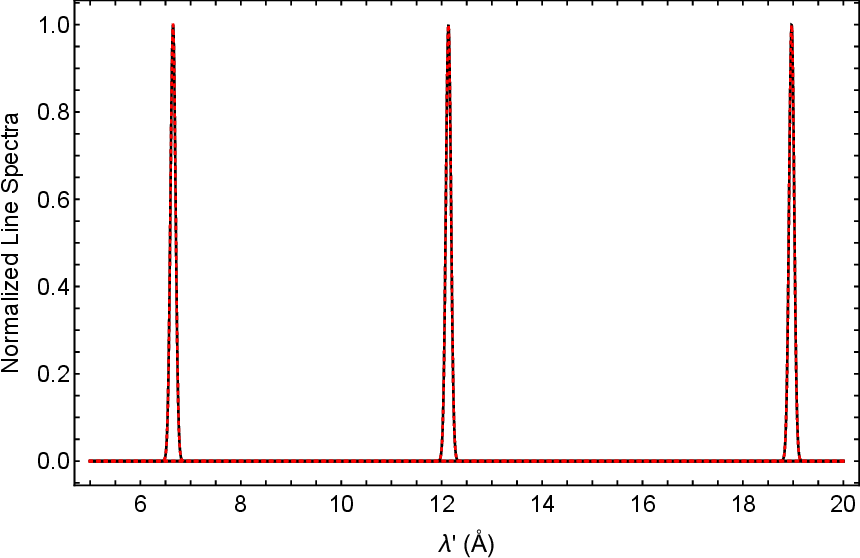}
    \end{minipage}
    \begin{minipage}{0.49\linewidth}
        \includegraphics[width=\linewidth]{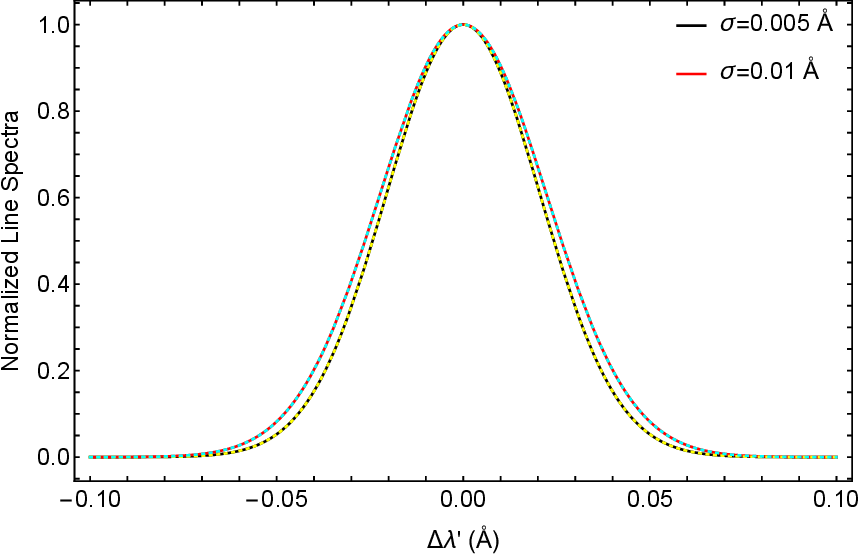}
    \end{minipage}
    \caption{Left: Comparison of normalized unfolded lines (solid black) and normalized counts lines (dotted red) at different wavelengths. Each line had a width of $\sigma_L=0.05$\,\AA\ while the detector response width was $\sigma=0.01$\,\AA. Right: Comparison of a normalized unfolded line (solid) and a normalized counts lines (dotted) for different detector response widths. The responses widths used correspond to \textit{Chandra's} HEG (black, yellow) and MEG (red, cyan). Both cases used the same line width of $\sigma_L=0.02$\,\AA.}
    \label{fig:GratingLineProfileError}
\end{figure*}

The next question is whether the grating's response width will cause any distortions. We show how the width affects the unfolding in the right panel of Figure~\ref{fig:GratingLineProfileError} for the case of a line with $\sigma_L=0.2$\,\AA\ width when viewed in a detector with the same response width as the \textit{Chandra} \Change{HETG High Energy Grating} (HEG; black, yellow) and MEG (red, cyan). The solid lines are the unfolded lines, which are again an exact match to the counts profiles (dotted lines).

We can conclude from the power law and Gaussian examples that unfolding causes minimal distortions in grating spectra under the following conditions
\begin{enumerate}
    \item The detector is high resolution,
    \item The continuum is not steeply sloped.
\end{enumerate}
Additionally, any local features like an emission line will not have their shapes distorted. These 
requirements are not well defined in terms of a quantifiable metric. Instead, individual detectors 
must be evaluated for their own unique conditions given their detector characteristics. For 
example, based on Figure~\ref{fig:GratingsPowerlawError}, we can conclude that for 
\textit{Chandra}'s HETG, which is senstive down to 1\,\AA, continuum slopes of up to $\gamma=6$ 
can be unfolded and be representative of the true source's slope to within 1 percent. On the other 
hand, a lower resolution detector sensitive at the same wavelengths would have more significant 
distortions to the continuum at the same wavelengths.

\subsection{Gaussian Response with Discontinuous Area}\label{sec:SharpFeature}

X-ray detectors often contain sharp features due to edges, dead pixels, gaps, and other physical 
attributes of their systems that make these toy models less representative of real situations. 
These sharp features can be smoothed over through dithering, but not all detectors apply such 
methods. \Change{Forward folding  accounts for spectrally sharp detector features in the transformation from smooth model to discrete counts, but here we will demonstrate how the 
unfolding defined by Equation~\eqref{eq:Sbar} can very effectively remove detector artifacts and provide cleaner spectra for guiding analysis.}

While there are a number of possible sharp features, we will use a simple example of a large 
discontinuity in the effective \Change{area}
\begin{equation}
    A(\lambda)=A_0\begin{cases}
        1 & \lambda \leq 10\,\text{\AA}\\
        5 & \lambda > 10\,\text{\AA}
    \end{cases}.
\end{equation}

If we use this area in the power-law source, the count rate spectra (in arbitrary units) is given 
as the blue curves in Figure~\ref{fig:PowerLawSharpFeature}. In both panels the power law index is 
$\gamma=1$ while the response widths are $\sigma=0.01$ (left) and $\sigma=0.1$ (right). In both 
cases, there is a clear jump in the spectra at the discontinuity without much influence in the 
width of the response. The unfolded spectra is plotted with arbitrary units in orange on the same 
plots and shows that the \Change{discontinuity} is completely removed. There are some complications 
from the unfolding being incorrect at the shortest wavelengths in the right panel, but those have 
been discussed already. What these plots highlight is that sharp features can be removed, 
independent of how high resolution the response is, without further distortions in broadband 
continuum sources.

\begin{figure*}
    \center
    \begin{minipage}{0.49\linewidth}
        \includegraphics[width=\linewidth]{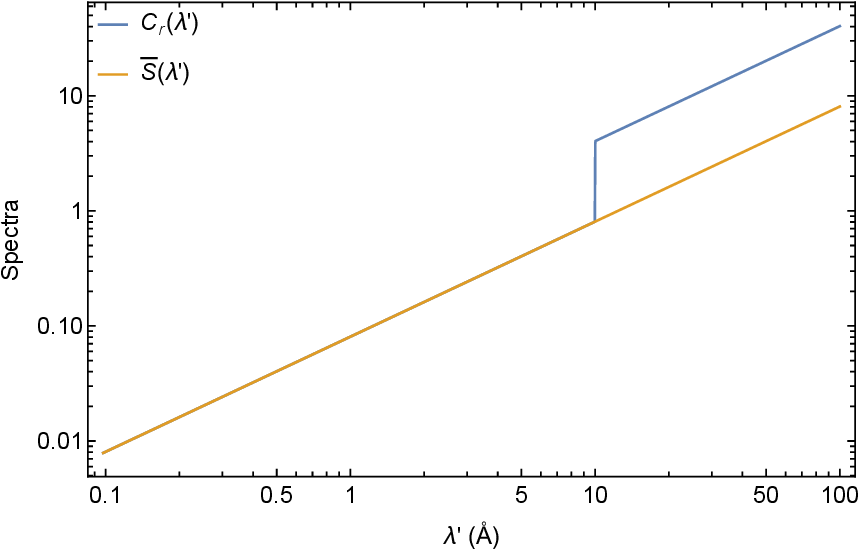}
    \end{minipage}
    \begin{minipage}{0.49\linewidth}
        \includegraphics[width=\linewidth]{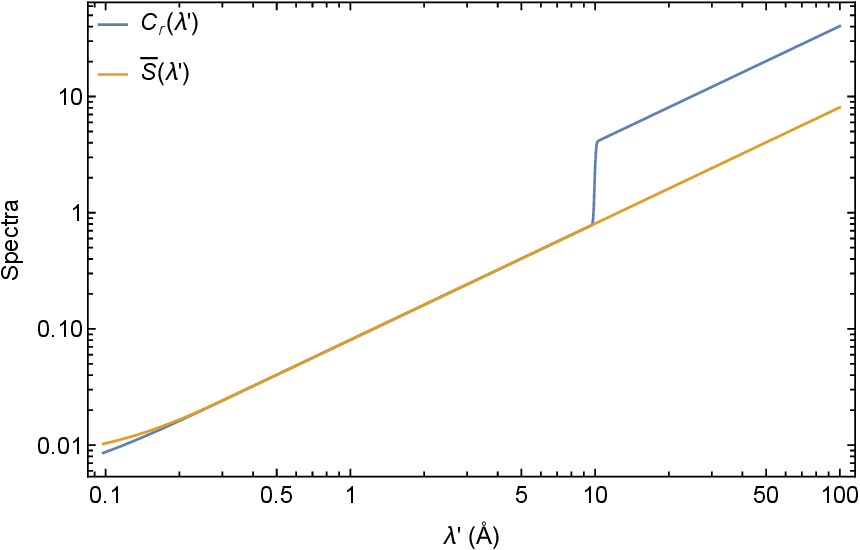}
    \end{minipage}
    \caption{\Change{Example of how unfolding can remove response features. Both panels use a $\gamma=1$ power law source model. Left panel uses $\sigma=0.01$\,\AA\ while the right uses $\sigma=0.1$\,\AA.  There is a discontinuity in the effective area at $10\,$\AA, seen in the count spectrum (blue), but effectively removed from the unfolded spectrum (orange).}}
    \label{fig:PowerLawSharpFeature}
\end{figure*}

Sharp feature can be more problematic for local source features, so we will next demonstrate how these are affected using the same Gaussian line profile from above. How a line is affected by the discontinuity in our effective area is dependent on where its line center is with respect to the discontinuity, either to the left/right or on top. The plots of these cases are given together in Figure~\ref{fig:LineSharpFeature} with three combinations of line and response widths. The top row is when both the line and response widths are equal, middle row is when the line is narrower than the response, and the bottom row is when the line is broader than the response. The left column is the count rate line profiles while the right column is the unfolded flux. \Change{All three lines have the same flux normalization, so they are meant to be identical except for their wavelengths and widths.} All plots are in arbitrary units.

\begin{figure*}
    \center
    \begin{minipage}{0.49\linewidth}
        \includegraphics[width=\linewidth]{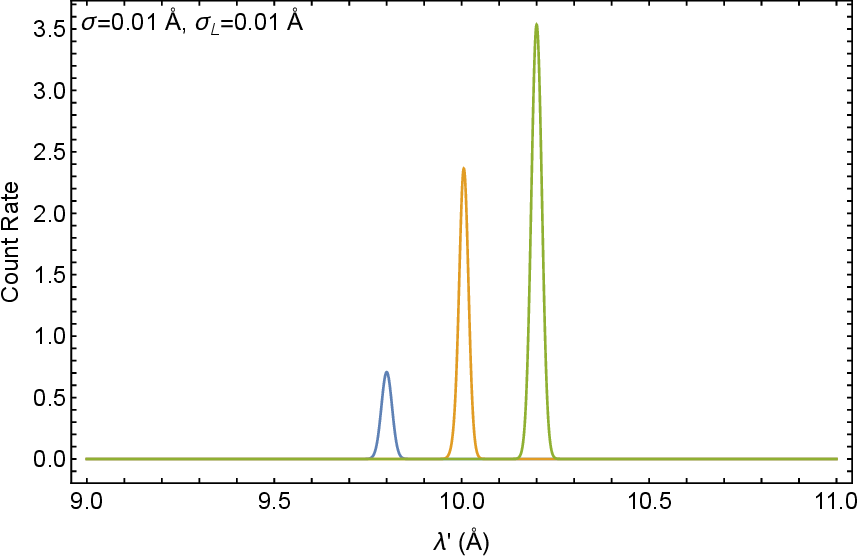}\\
        \includegraphics[width=\linewidth]{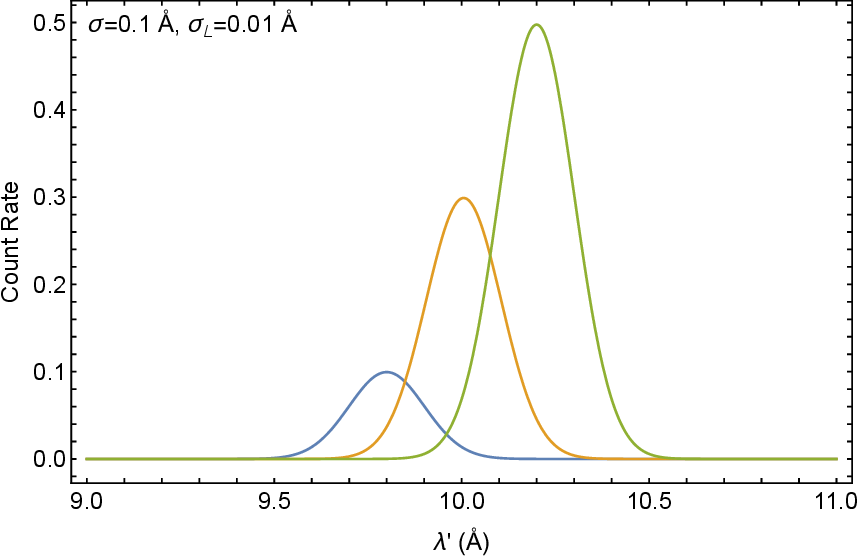}\\
        \includegraphics[width=\linewidth]{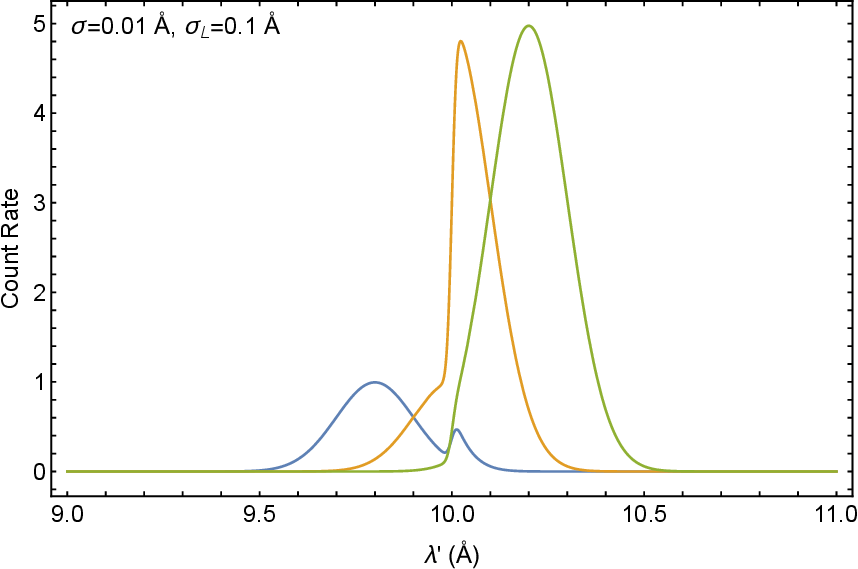}
    \end{minipage}
    \begin{minipage}{0.49\linewidth}
        \includegraphics[width=\linewidth]{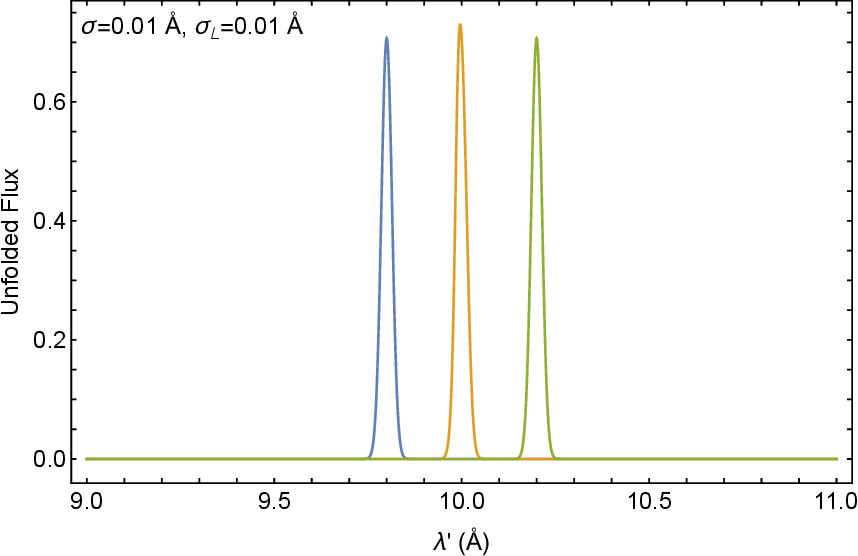}\\
        \includegraphics[width=\linewidth]{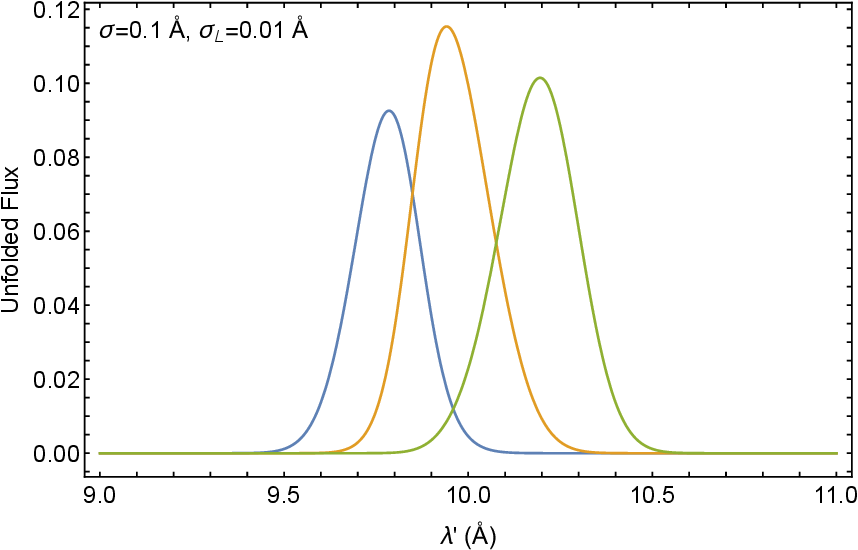}\\
        \includegraphics[width=\linewidth]{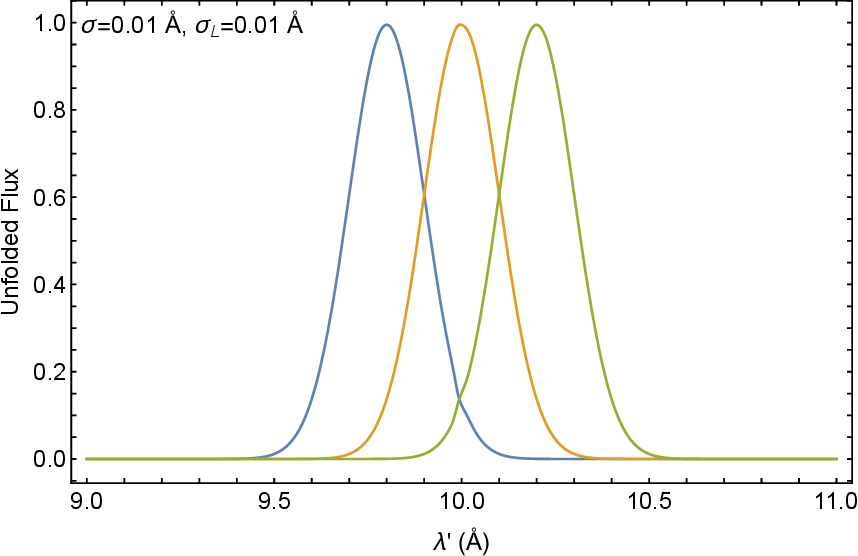}
    \end{minipage}
    \caption{Examples of unfolding local spectral features, e.g. emission lines, across a sharp response features. Left column is the observed count rates of a Gaussian line profile while the right column is the unfolded version. \Change{The three lines have the same flux. There is a discontinuity in the effective area at $10\,$\AA}}
    \label{fig:LineSharpFeature}
\end{figure*}

First we need to note the two kinds of distortions that the discontinuity causes in the left column. One is a simple change in the apparent relative strength between the three lines in each row due to the higher effective area collecting more counts. The other is a more complex distortion in line shape seen in the bottom row. For a line to the left of the sharp feature, a ``ghost line'' is formed at the jump whereas the line to the right has a sharper left wing. The line centered on of the jump is severely distorted with the entire left half being truncated. So for our unfolding process to be accurate, we need these to be corrected to a normal Gaussian shape and the relative strengths to be equalized.

Looking at the right column of Figure~\ref{fig:LineSharpFeature}, we can see that the unfolding is again dependent on both line and response widths. When the widths are equal, the relative strengths of the lines are reasonably recreated. There is a minor distortion in the line centered on the discontinuity, but it is not significant. When the lines are narrower than the response, there is still a significant difference in the line strengths. \Change{Which can be attributed to the lines not being resolved by the detector.} Of note for this case is the severity of the distortion of the lines. We again find that the central line (orange) centered on the discontinuity is most affected, in this case appearing to contain a significantly higher flux than the rest of the lines. On the other hand, while the lines to the left and right of the discontinuity show similar line strengths, they have visible difference in their unfolded fluxes.  Finally, when the line is broader than the response, both the line strengths and shape distortions are unfolded with only minor distortions to their overall shape. The minor distortions are primarily in the wings of the profiles to the left and right of the discontinuity. Such minor distortions would be at the level of noise in a real detector, an effect not demonstrated with these toy models, but would be modelled appropriately in any rigorous analysis.

\subsection{CCD Detectors}

For the case of CCD detectors, we will follow a similar principle with defining simple response functions with constant area. The response function of a CCD detector has many features beyond the primary photopeaks for a given incident photon \citep{Yacout1986} that make these demonstrations difficult. As such, we will restrict our toy model to that of only the Gaussian-like primary photopeak, \Change{which \citet{Yacout1986} parameterized as}
\begin{equation}
    R(E',E) = \frac{R_0}{\sqrt{2\pi}\sigma(E)}\exp\left(-\frac{1}{2}\left(\frac{E'-E}{w \sigma(E)}\right)^2\right).
\end{equation}
The response width of a CCD is energy dependent, with the functional form
\begin{equation}
    \sigma^2(E) = F\frac{E}{w}+\sigma_\mathrm{read}^2,
\end{equation}
where $F$ is the Fano factor, $E/w$ is the number of electrons produced by an incident photons, \Change{and $E$ and $E'$ are the photon and detector bin energies respectively}. The resolution of the detector is broadened by the readout noise $\sigma_\mathrm{read}$. For the examples that will follow below, we will use numbers that are relevant for \textit{Chandra}'s ACIS-I detector, so unless stated otherwise we will use $1/w=1/3.7$\;$e^{-}$\,eV$^{-1}$, $F=0.12$, and $\sigma=2\,e^-$ \Change{\citep{Arnaud2011}}.

If we start with the case of a power-law continuum source again, we first must note that there is no closed-form solution to the unfolded spectra in this case, so we will go directly to the plotted unfolding errors in Figure~\ref{fig:CCDPowerlawError}. We can note that for CCD detectors the unfolding's accuracy is dependent on the continuum's slope in the same way as the grating cases. Spectra with steeper slopes, as inferred by the simple power laws used here, will have less accurate unfoldings. Because of the energy dependence of a CCD detector's response width, the error increase quickly for steeper continua. We can generally say that for $E' \geq 1$\,keV the unfolding has minimal error regardless of the continuum slope. On the other hand, the softer portion of the spectrum must be judged based on its slope as to whether it can be unfolded to within some predefined tolerance.

\begin{figure}
    \center
    \includegraphics[width=\linewidth]{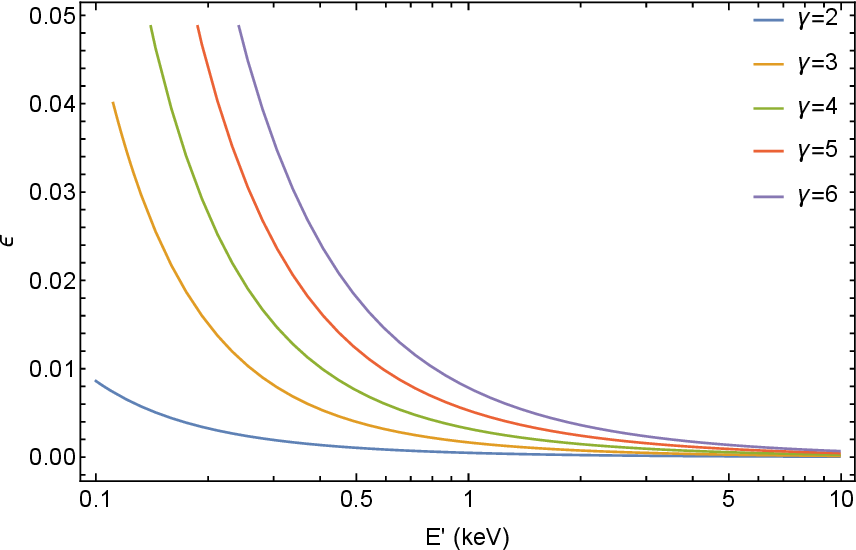}
    \caption{Error in the unfolding of a power law continuum observed with a CCD detector and constant area.}
    \label{fig:CCDPowerlawError}
\end{figure}

Moving onto the case of a line profile in Figure~\ref{fig:CCDLineProfileError}, we see again similarities to the grating case. For lines across the ACIS-I detector sensitivity, all can be similarly unfolded without any distortions. When the line width is increased, though, there can be minor distortions as shown for the case of $\sigma_L = 0.05$\,keV in the right panel. The unfolded line (solid red) is slightly redshifted as compared to the corresponding counts profile (dotted cyan). This distortion is exceptionally minor and would be covered by uncertainties in real data, but highlights the potential for measurable distortions if the resolution is low.

\begin{figure*}
    \center
    \begin{minipage}{0.49\linewidth}
        \includegraphics[width=\linewidth]{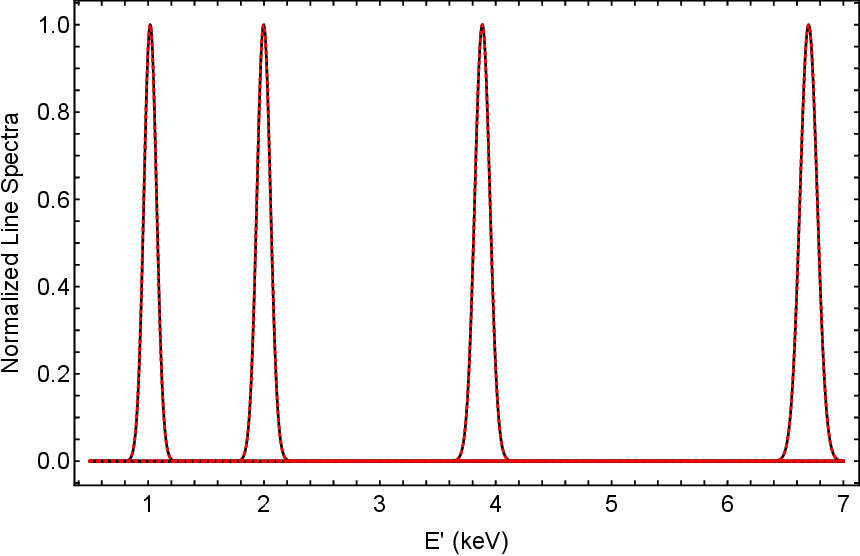}
    \end{minipage}
    \begin{minipage}{0.49\linewidth}
        \includegraphics[width=\linewidth]{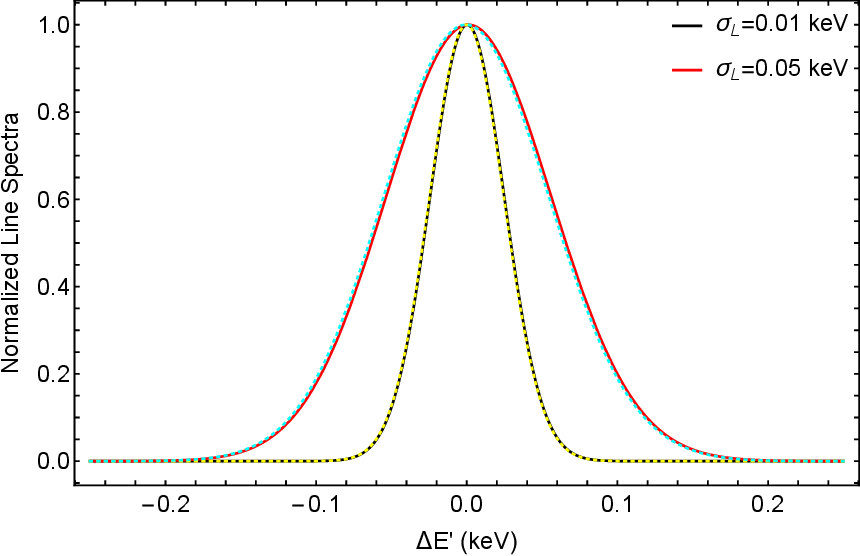}
    \end{minipage}
    \caption{\textbf{Left:} Comparison of normalized unfolded lines (solid black) and normalized counts lines (dotted red) at different wavelengths. Each line had a width of $\sigma_L=0.05$\,\AA. \textbf{Right:} Comparison of a normalized unfolded lines (solid) and normalized count lines (dotted) for different line widths.}
    \label{fig:CCDLineProfileError}
\end{figure*}

This toy CCD model allows to make the same general statements about unfolding spectra. If the detector is high resolution and the continuum is not steeply sloped, then the unfolded spectra will provide an accurate representation of the source for a given observation. Local features are potentially subject to these conditions as well but in general show distortions within ususal X-ray observation uncertainties. Unlike the gratings, there is an additional caveat that the unfolding is more energy dependent like a CCD response width. A given observation may only be unfoldable within a specific energy range where both conditions are met.

\section{Real Detectors: Example Unfolding Uses}\label{sec:RealDetectors}

In this section we will demonstrate specific use cases for unfolding real X-ray data from a number of telescopes. The list of objects and their Obs Ids are given in Table~\ref{tab:ObsIds}. \textit{Chandra} data was reprocessed with the standard pipeline in \textsc{ciao} version 4.15 \citep{Fruscione2006}. \Change{\textit{NuSTAR} data was processed with the standard pipeline in the \textsc{heasoft} version 6.33 tools provided by \textsc{sciserver} \citep{SciServer2020}.} Finally, for \textit{XMM}, we used the data products provided by the Pipeline Processing System. All figures generated from this data were made using \textsc{ISIS} \Change{version 1.6.2-51}.

\begin{deluxetable*}{llllc}
    \tablecaption{Observation IDs. \label{tab:ObsIds}}
    \tablehead{
        \colhead{Object} & \colhead{Telescope} & \colhead{Obs Id} & \colhead{Start Date} & \colhead{Exposure time (ks)}
    }
    \startdata
        \object{3C 273} & \textit{Chandra} & 459 & 2000-01-10 & 38.67 \\
        \object{$\zeta$ Pup} & \textit{XMM-Newton} & 0095810301& 2000-06-08& 57.42\\
        &&0095810401& 2000-10-15& 40.61\\
        &&0157160401& 2002-11-10& 42.41\\
        &&0157160501& 2002-11-17& 42.53\\
        &&0157160901& 2002-11-24& 43.67\\
        &&0157161101& 2002-12-15& 39.01\\
        &&0159360101& 2003-05-30& 69.34\\
        &&0163360201& 2003-12-06& 53.68\\
        &&0159360301& 2004-04-12& 61.34\\
        &&0159360401& 2004-11-13& 63.10\\
        &&0159360501& 2005-04-16& 64.21\\
        &&0159360701& 2005-10-15& 30.12\\
        &&0159360901& 2005-12-03& 53.58\\
        &&0159361101& 2006-04-17& 52.95\\
        &&0414400101& 2007-04-09& 63.91\\
        &&0159361301& 2008-10-13& 61.50\\
        &&0561380101& 2009-11-03& 64.31\\
        &&0561380201& 2010-10-07& 76.91\\
        \object{$\pi$ Aqr} & \textit{Chandra} & 26079 & 2022-08-23 & 9.93 \\
        &  & 27269 & 2022-08-27 & 9.99 \\
        &  & 26080 & 2022-09-05 & 29.67 \\
        &  & 26001 & 2022-09-13 & 19.80 \\
        &  & 27412 & 2022-09-14 & 21.78\\
        &  & 27325 & 2022-10-30 & 9.76\\
        & \textit{NuSTAR} & 30501009002 & 2019-11-04 & 49.93
    \enddata
\end{deluxetable*}

We provide no true analysis of the observations that are presented below. Our intention is to illustrate only, so any interpretations given should to not be taken beyond the scope of this text. Any interesting features will be left to future works by other authors to analyze.

\subsection{Correcting Detector Features}

The first demonstration is the removal of spectral features caused by the detector. We showed in \S~\ref{sec:SharpFeature} how a sharp feature like a jump discontinuity in the effective area can be removed by unfolding in theory. Real detectors are much more complex and deserve a full demonstration. Our example target for this is the quasar 3C\,273, whose early \textit{Chandra} \Change{HETG, specifically the MEG,} observation is shown in Figure~\ref{fig:3C273}. Our focus is on the 3 -- 9\,\AA\ region, where the Si absorption edge occurs near 7\,\AA. The left panel shows the usual count rate spectrum (black histogram) and the scaled MEG ARF (red line). In this region, the spectrum of 3C 273 follows the ARF of the spectrum perfectly, especially the Si edge, making \Change{it is difficult to infer an appropriate model to forward fold.}

One would of course use a priori knowledge of quasar physics to determine a suitable model and extract useful quantities. But let us imagine a scenario where such a priori knowledge is not known. Under such conditions, how would one determine an appropriate model? Unfolding can provide insights for making this determinination. In the right panel of Figure~\ref{fig:3C273}, we plot the scaled ARF again but this time with the unfolded spectrum of 3C\,273. The actual spectral shape of the object is \Change{clearer in this plot}, looking like either a linear line or a shallow power law. Proper forward folding techniques should still be applied as this only provides a \Change{guide to the} appropriate model to use.

\begin{figure*}
    \begin{minipage}{0.49\linewidth}
        \includegraphics[width=0.7\linewidth, angle=-90, trim=0 40 0 0]{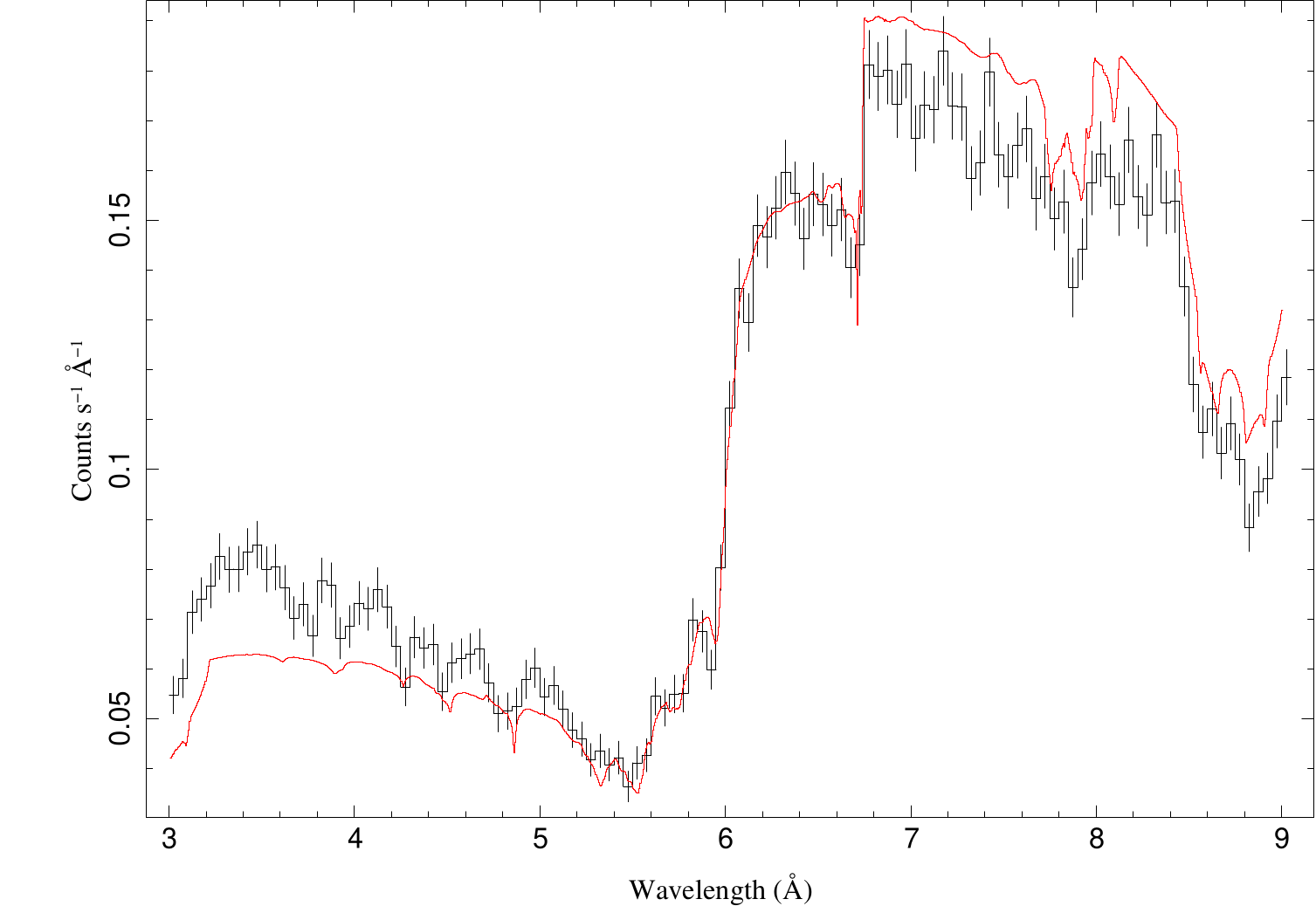}
    \end{minipage}
    \begin{minipage}{0.49\linewidth}
        \includegraphics[width=0.7\linewidth, angle=-90, trim=40 40 0 0]{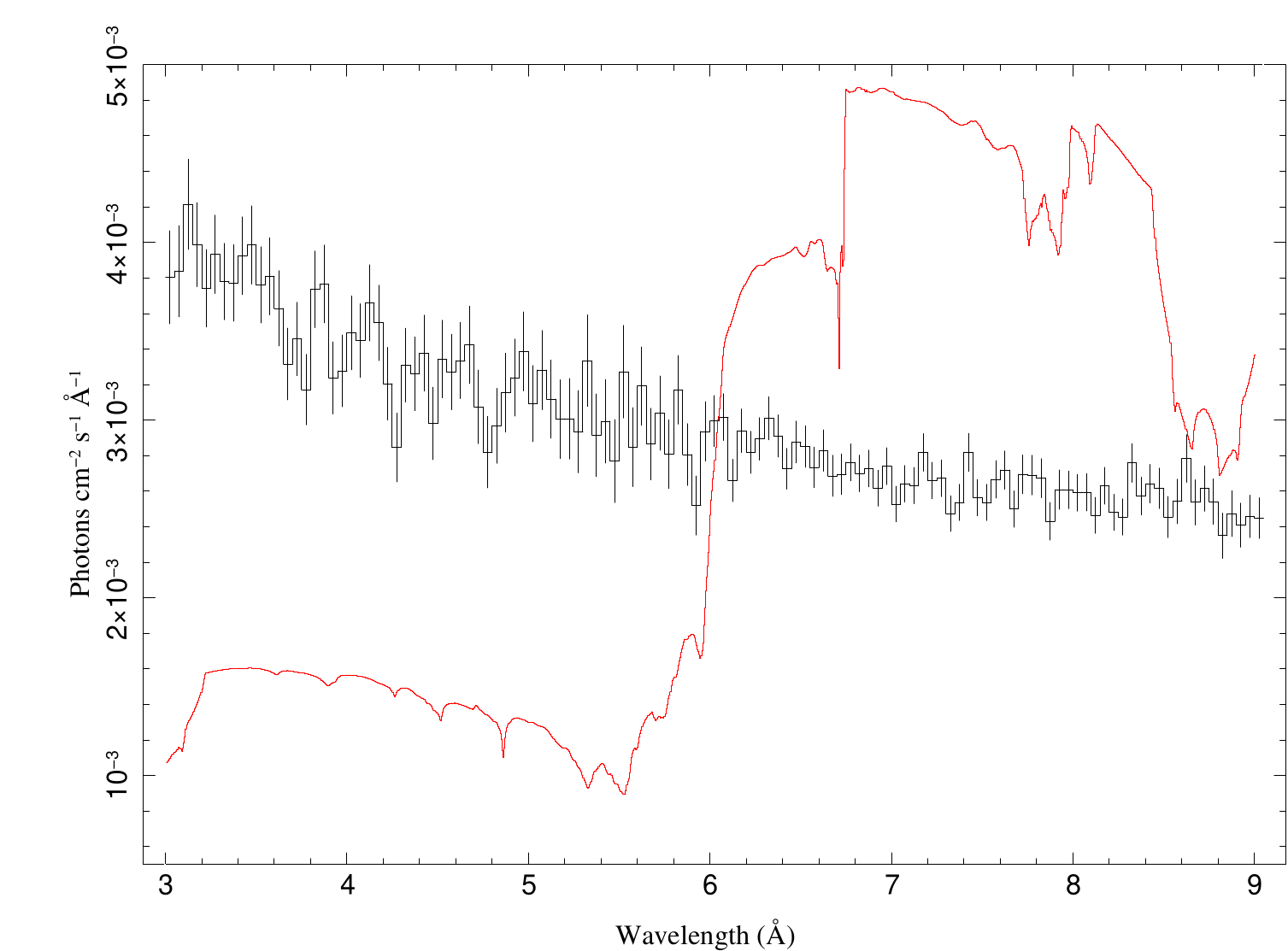}
    \end{minipage}
    \caption{\textit{Chandra} MEG observation of 3C 273 showing of how unfolding can remove features in spectra caused by the detector. Left panel is the counts spectra while the right panel is the unfolded version. In both panels, the black histogram is the spectrum binned by a constant factor of 10 and the scaled ARF values are in red.}
    \label{fig:3C273}
\end{figure*}

Of course not all observations will be this clear after unfolding, so unfolding can also be used for presentation purposes. \Change{This is shown in Figure~\ref{fig:zetaPup}, where we plot the \ion{Fe}{17} line region at 15\,\AA\ of 18 combined observations of $\zeta$ Pup taken with \textit{XMM}'s Reflection Grating Spectrometer {RGS}. We specifically plot only the RGS1 data to show the drop outs from detector edges in the count rate spectrum caused by \textit{XMM} not dithering, which would smear them over. Unfolding this spectra corrects these drop outs, providing a cleaner spectrum for presenting and an accurate representation of the intrinsic source spectrum.}

\begin{figure}
    \includegraphics[width=\linewidth]{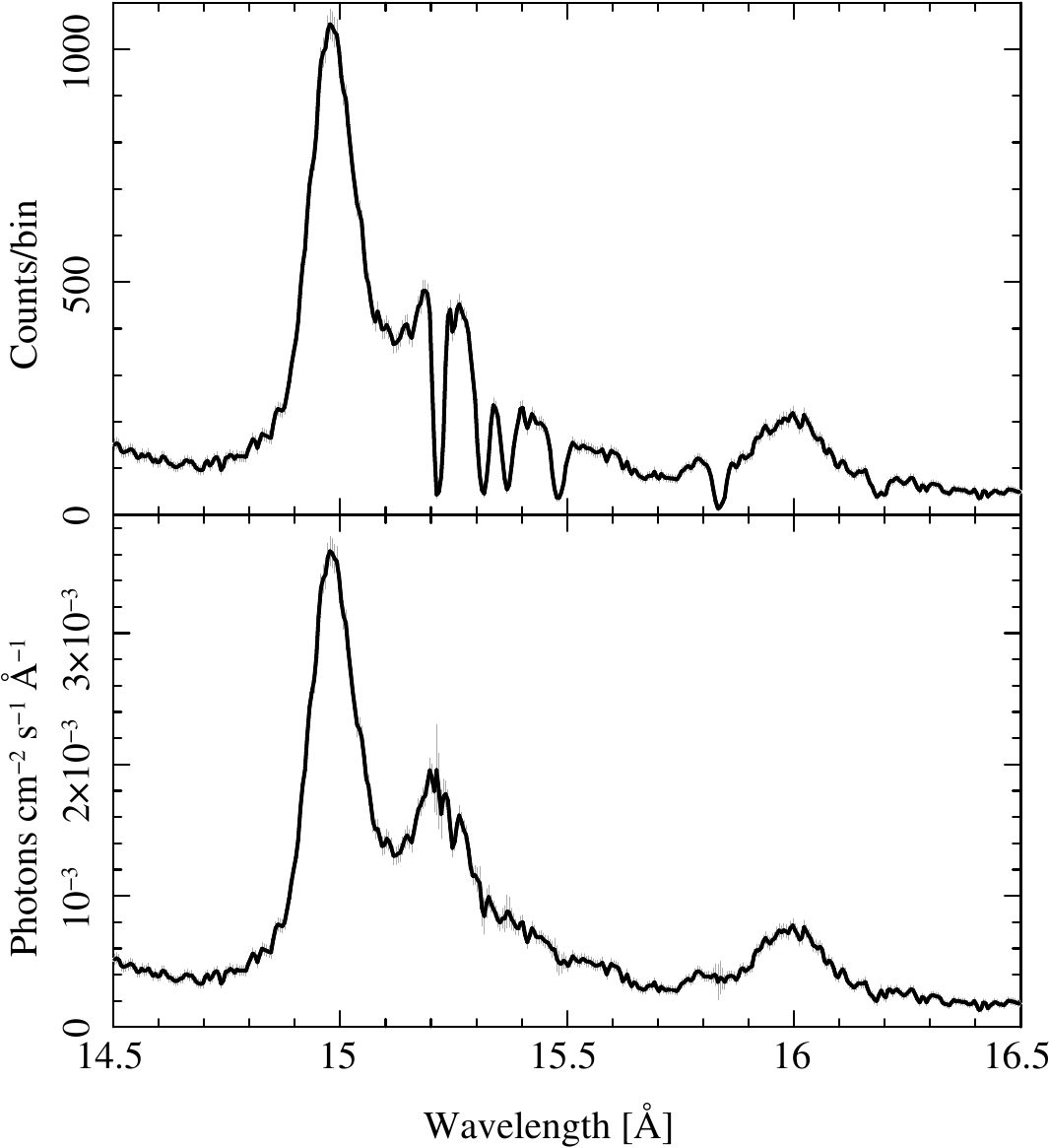}
    \caption{\Change{This combination of 18 XMM/RGS-1 observations, nearly 1 Ms exposure of $\zeta$ Puppis, demonstrates how drop-outs in the count spectrum---the prominent dips in the upper panel---which are due to sharp features in the effective area function, are corrected by unfolding (bottom). Statistical errorbars are plotted in gray, but are hard to see, except in the regions of very low counts.  The bottom panel is thus an accurate representation of the intrinsic source spectrum.  The errors are not magnified by unfolding, but are simply rescaled and have the same fractional uncertainty as in the upper counts spectrum.}}
    \label{fig:zetaPup}
\end{figure}

\Change{It should be made clear at this point that the process of unfolding does not change the relative uncertainty in a given spectrum. While the errors appear magnified in the bottom plot of Figure~\ref{fig:zetaPup}, they still represent the same fractional uncertainty as in the top plot. Since unfolding is only a simple rescaling in Equation~\eqref{eq:Sbar} and response functions are assumed to not have any intrinsic error, the uncertainty in a given channel in
\begin{equation}
    \delta \bar{S}(\lambda') = \frac{\delta C(\lambda')}{t\int_{\Delta \lambda(\lambda')} \mathscr{R}(\lambda', \lambda)\dd\lambda},
\end{equation}
where $\delta C(\lambda')$ is the uncertainty in the counts in the bin $\lambda'$.
}

\subsection{Comparing Detectors: Chandra and NuSTAR}

The more powerful use of unfolding is the ability to directly compare observations of the same object with different telescopes. Due to differences in effective areas, the usual presentation of count rate spectra is not always easily interpretable. As an example, see the top plot of Figure~\ref{fig:piAqr} showing \textit{Chandra} HEG + MEG (black histogram) and \textit{NuSTAR} A + B (blue histogram) observations of $\pi$ Aquarii. There is a large difference in the observed count rates because of the difference between the \textit{NuSTAR} and \textit{Chandra} response and effective area.

\begin{figure*}
    \centering
    \includegraphics[width=0.4\linewidth,angle=-90,trim=0 40 0 0]{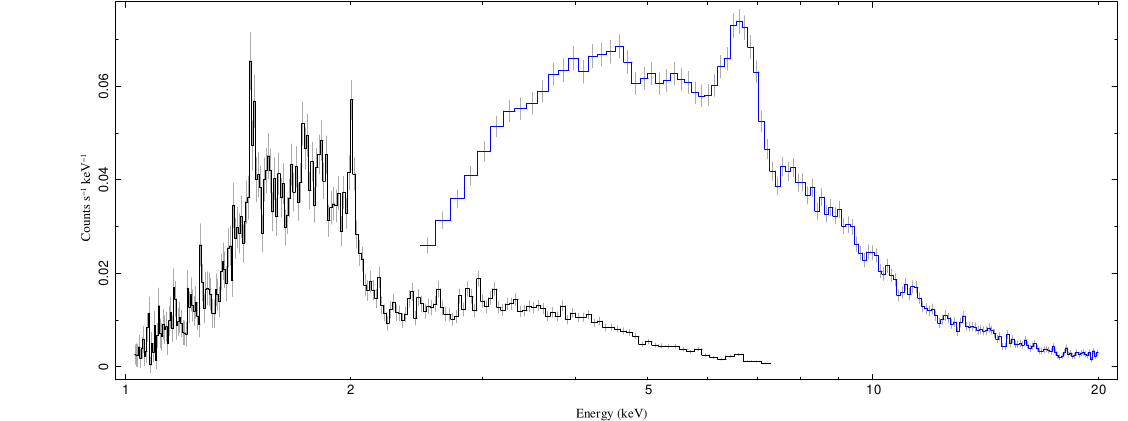}
    \includegraphics[width=0.4\linewidth,angle=-90,trim=0 40 0 0]{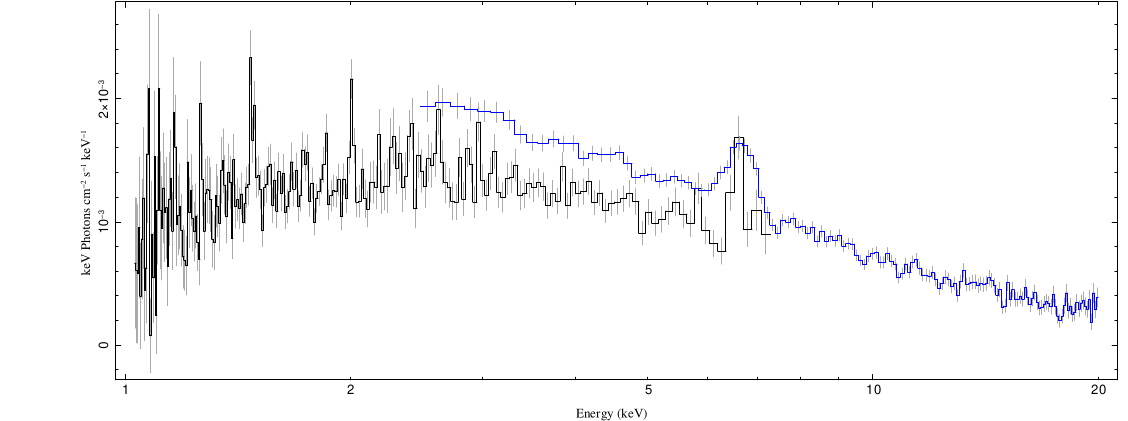}
    \caption{Count rate (top) and unfolded (bottom) spectra of $\pi$\,Aqr from \textit{Chandra}'s HETG (black) and \textit{NuSTAR} (blue). The HETG spectrum is the HEG + MEG plotted on the MEG grid and binned by a constant factor of 10. The \textit{NuSTAR} spectrum is telescopes A + B plotted on A's grid and binned by a constant factor of 3.}
    \label{fig:piAqr}
\end{figure*}

If we unfold the spectra, we can remove this difference and make a direct comparison \Change{of the intrinsic flux of the $\pi$\,Aqr}. This is shown in the bottom panel of Figure~\ref{fig:piAqr}. We can see now that there is an apparent difference in the observed flux between observations. \textit{NuSTAR} and \textit{Chandra} are known to be very close in flux calibration \citep{Madsen2017}, so this difference is real. This is also consistant with the known variability of $\pi$ Aqr \citep{Naze2019,Huenemoerder2024}.

\section{Conclusion}\label{sec:Conclusions}

We showed the conditions for which the ``unfolding'' method defined by Equation~\ref{eq:Sbar} gives an accurate representation of the source flux spectrum. This method, implemented within the \textsc{isis} software, distinguishes itself from other methods by being model independent. The resulting flux spectra is proven to be a unique representation for a given observation, but the uncertainty in the representation is dependent on what one assumes the actual source's functional form is. As such, this unfolding procedure is not a deconvolution method since it does not assume nor solve for any source flux functional form.

The conditions for accurate unfolding, inferred from toy models of grating and imaging CCD detectors, are the following:
\begin{enumerate}
    \item The detector is high-resolution,
    \item The spectrum is not steeply sloped.
\end{enumerate}
These conditions apply to both a grating and CCD detector, although the latter is more dependent on the correlation between the two. There are some cases, in particular based on the energies of interest, for which a spectrum can be very steeply sloped by still be unfolded. At the same time, the reverse is also true as there can be cases of a region with a shallow slope and lower resolution that can be unfolded accurately. Thus, each detector and spectrum must be judged individually to determine if they fit into these conditions and how much error is okay to introduce through the unfolding. In all cases, however, local source features like emission lines can be unfolded with \Change{minimal distortion to the intrinsic spectral shapes}. There is minimal to no distortion to the line shapes beyond the intrinsic broadening from the response that can not be removed.

We also showed cases for ways in which unfolding can be used for analysis and presentation. Since the unfolding removes the influence of the response, we can remove features present in spectrum caused by the effective area like absorption edges due to the detector material. Such uses can help elucidate appropriate models for analysis but primarily provide cleaner presentations of data. Unfolding is particularly useful when working with data from different telescopes on the same object. When cross-calibration factors are accounted for, the unfolded spectra show the flux difference either across different wavelengths or through time.

For current generation telescopes like \textit{Chandra}, \textit{XMM}, and \textit{NuSTAR}, their spectral resolving powers are high enough that unfolded spectra can be used for the purposes demonstrated here. With next generation detectors (\textit{XRISM} \citep{Makoto2020}, \textit{NewAthena} \citep{Bavdaz2023}, \textit{ARCUS} \citep{Smith16}), however, the available resolving power pushes closer to a Delta functional and makes the unfolding even more accurate. It may be the case that in future observations, the unfolded flux spectra could be fit directly instead of through forward folding.

\begin{acknowledgments}
    Support for S.J.G. and D.P.H. was provided by NASA through the Smithsonian Astrophysical Observatory (SAO) contract SV3-73016 to MIT for support of the Chandra X-Ray Center (CXC) and Science Instruments. CXC is operated by SAO for and on behalf of NASA under contract NAS8-03060. This research has made use of ISIS functions (ISISscripts) provided by ECAP/Remeis observatory and MIT (\url{http://www.sternwarte.uni-erlangen.de/isis/}).

    \Change{This paper employs a list of Chandra datasets, obtained by the Chandra X-ray Observatory, contained in~\dataset[DOI: 10.25574/cdc.299]{https://doi.org/10.25574/cdc.299}.}

    \facilities{\textit{XMM-Newton} (RGS), \textit{CXO} (HETG/ACIS-S), \textit{NuSTAR}}
    \software{ISIS \citep{Houck2000}}

\end{acknowledgments}

\bibliography{bib}{}
\bibliographystyle{aasjournal}




\end{document}